\documentclass[conference]{IEEEtran}
\IEEEoverridecommandlockouts
\usepackage{algorithm}
\usepackage{algorithmic}
\usepackage{amsfonts}
\usepackage{amsmath}
\usepackage{amssymb}
\usepackage{amsthm}
\usepackage[f]{esvect}
\usepackage{graphicx}
\usepackage{hyperref}
\usepackage[utf8]{inputenc}
\usepackage{subfig}

\newtheorem{definition}{Definition}[section]
\newtheorem{theorem}{Theorem}[section]

\newtheorem{example}{Example}[section]

\newcommand{\cost}{\mathit{cost}}
\newcommand{\deposit}{\mathit{deposit}}
\newcommand{\compensation}{\mathit{comp}}

\newcommand{\FairSwapDeploymentCostDate}{2021-12-17}
\newcommand{\FairSwapTestBlock}{13,823,842} 
\newcommand{\FairSwapDeploymentBaseGasPrice}{60}
\newcommand{\FairSwapDeploymentExchangeRate}{3880}
\newcommand{\FairSwapDeploymentCost}{349.20}

\begin{document}

\title{Formalizing Cost Fairness for Two-Party Exchange Protocols using Game Theory\\and Applications to Blockchain (Extended Version) \\ \thanks{Kenneth Skiba was supported by the Deutsche Forschungsgemeinschaft under grant KE 1413/11-1 and Jan Jürjens by the EC (Horizon 2020) within the projects "Digital Reality in Zero Defect Manufacturing (Qu4lity)" and "Trusted Secure Data Sharing Space (TRUSTS)".}}

\author{
	\IEEEauthorblockN{
		Matthias Lohr\IEEEauthorrefmark{1},
		Kenneth Skiba\IEEEauthorrefmark{2},
		Marco Konersmann\IEEEauthorrefmark{1},
		Jan Jürjens\IEEEauthorrefmark{1}\IEEEauthorrefmark{3},
		Steffen Staab\IEEEauthorrefmark{4}\IEEEauthorrefmark{5}
	}
	\IEEEauthorblockA{\IEEEauthorrefmark{1}Institute for Software Technology, University of Koblenz-Landau, Koblenz, Germany}
	\IEEEauthorblockA{\IEEEauthorrefmark{2}Artificial Intelligence Group, Fernuniversität in Hagen, Hagen, Germany}
	\IEEEauthorblockA{\IEEEauthorrefmark{3}Fraunhofer ISST, Dortmund, Germany}
	\IEEEauthorblockA{\IEEEauthorrefmark{4}Institute for Parallel and Distributed Systems (IPVS), University of Stuttgart, Stuttgart, Germany}
	\IEEEauthorblockA{\IEEEauthorrefmark{5}University of Southampton, Southampton, United Kingdom}
}

\maketitle

\begin{abstract}
Existing fair exchange protocols usually neglect consideration of cost when assessing their fairness.
However, in an environment with non-negligible transaction cost, e.g., public blockchains, high or unexpected transaction cost might be an obstacle for wide-spread adoption of fair exchange protocols in business applications.
For example, as of \FairSwapDeploymentCostDate, the initialization of the FairSwap protocol on the Ethereum blockchain requires the selling party to pay a fee of approx.\ \FairSwapDeploymentCost~USD per exchange.
We address this issue by defining cost fairness, which can be used to assess two-party exchange protocols including implied transaction cost.
We show that in an environment with non-negligible transaction cost where one party has to initialize the exchange protocol and the other party can leave the exchange at any time cost fairness cannot be achieved.
\end{abstract}


\section{Introduction}
\label{sec:introduction}

In commerce, two or more parties want to exchange goods.
According to Asokan \cite{DBLP:books/daglib/0094910}, an exchange becomes a \emph{fair exchange} iff it is guaranteed that either all involved parties get exactly the good they requested, or no good has been transferred at the end of the exchange \cite{DBLP:books/daglib/0094910,DBLP:journals/cj/PagniaVG03}.
It has been shown that a trusted third party is required to achieve fairness for a two-party exchange \cite{even1980relations,pagnia1999impossibility}.
In non-digital exchanges (e.g., buying/selling a house), notaries or banks take on the role of a trusted third party.
In electronic commerce, several approaches have been developed that ensure a fair exchange between two parties either utilizing dedicated organizations as trusted third parties or utilizing blockchains (or more general, distributed ledgers) as distributed trusted third party \cite{DBLP:journals/iacr/Delgado-SeguraP17,DBLP:conf/ccs/DziembowskiEF18,DBLP:conf/ccs/EckeyFS20,wagner2019dispute,hall2019fastswap}.

When a trusted third party is involved in an exchange, it can raise non-negligible transaction cost (e.g., notary fees or fees for a bank guarantee).
Such transaction cost must be considered separately from possible payments as part of the exchange, as they are intended to pay the trusted third party for their services rather then being part of the goods (including money) to be exchanged between the participants\footnote{In this work, we use the terms \emph{transaction} and \emph{transaction cost} generally for interactions with the trusted third party and resulting cost.}.

When an exchange protocol is used in which a public blockchain (e.g., Ethereum \cite{wood2014ethereum}) acts as a trusted third party, all interactions with the trusted third party are performed using \emph{blockchain transactions}, which require the acting party to pay transaction cost in form of \emph{blockchain transaction fees}\footnote{Every time we need to refer to concrete type of transaction or transaction cost, e.g., in context of blockchains, we prefix it with the according concretization, such as \emph{blockchain transaction} and \emph{blockchain transaction fee}}.
For example, the initialization of the FairSwap protocol (deployment of a single-use smart contract for the exchange), which provides functionality to fairly sell data for money on the Ethereum blockchain, requires the selling party to pay for blockchain transaction fees of approx.\ 1,050,000 Gas\footnote{As stated by Dziembowski et al.\ \cite{DBLP:conf/ccs/DziembowskiEF18}. During our tests with minor bug fixes we observed cost of approx.\ 1,500,000 Gas. Our version of the smart contract with bug fixes is available online at \url{https://gitlab.com/MatthiasLohr/bdtsim}.}, which, as of \FairSwapDeploymentCostDate, is worth approx.\ \FairSwapDeploymentCost~USD\footnote{As of \FairSwapDeploymentCostDate, Ethereum block \FairSwapTestBlock~ was created with a base Gas price of approx.\ \FairSwapDeploymentBaseGasPrice~GWei/Gas and an exchange rate of approx.\ \FairSwapDeploymentExchangeRate~USD/Eth ($1$ Eth = $10^9$ GWei), which results in blockchain transaction fees of approx.\ \FairSwapDeploymentCost~USD for deployment the smart contract, assuming zero tip \cite{EIP1559}.}.
There exist alternative approaches, such as optimistic protocol design \cite{DBLP:conf/ctrsa/KupcuL10} or the usage of state channels  \cite{DBLP:conf/ccs/DziembowskiFH18} that can generally be used to reduce blockchain transaction fees.
Nevertheless, even then transaction cost is greater than zero and often non-negligible.

For private blockchains, the existence of transaction cost depends on the selected concepts and implementations decided to be applied.
E.g., the Hyperledger Fabric \cite{DBLP:conf/eurosys/AndroulakiBBCCC18} blockchain framework does per default not include any means or features of financial values or currencies.
However, also operation of a private blockchain costs money (e.g., for buying the required servers), which can be apportioned to each blockchain transaction sent to the private blockchain instance, or asking for a fixed monthly fee but not charging per blockchain transaction.

So far, all blockchain-based fair exchange protocols known to us only consider the whereabouts of the goods to be exchanged for fairness assessment, while they ignore transaction cost accrued by using the blockchain as trusted third party.
This opens the possibility for a \emph{grieving attack} \cite{DBLP:conf/ccs/EckeyFS20} as it is shown in Figure \ref{fig:grieving attack sequence diagram}, where an unfaithful party $B$ causes a faithful party $A$ to initiate an exchange with a transaction that accrues transaction cost and then leaves without finishing the exchange.
Doing so, an attacker can harm the attacked party (e.g., business opponent) with only low or even zero cost for the attacker while the attacked party has to bear possibly non-negligible transaction cost for the initialization.
Due to blockchain anonymity\footnote{It has been shown by, e.g., Biryukov and Tikhomirov that there exist several but unreliable methods for identity deanonymization on blockchains such as Bitcoin \cite{DBLP:conf/eurosp/BiryukovT19}. We assume that deanonymization might not be sufficiently reliable to prevent grieving attacks.}, the faithful party cannot reliably distinguish between a repeated request from the same unfaithful party or a new party.
Even given an exchange that is proven to be fair following the definition by Asokan \cite[p.~9f]{DBLP:books/daglib/0094910}, a faithful party may either accept incoming requests and risk bearing the costs of a grieving attack, or not accept incoming requests at all and thus not complete their planned exchange of goods.

\begin{figure}
	\centering
	\subfloat{%
		\includegraphics[width=.65\linewidth]{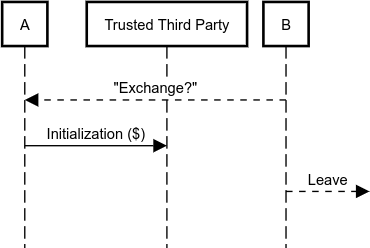}
	}\\
	\subfloat{%
		\includegraphics[width=.25\linewidth]{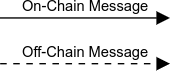}
	}
	\caption{Diagram of exemplary action sequence for a grieving attack, conducted by $B$. \emph{Initialization} is an action where $A$ pays fees to the Trusted Third Party in the belief that $B$ will continue the targeted exchange.}
	\label{fig:grieving attack sequence diagram}
\end{figure}

This raises the question of what an exchange protocol has to achieve in order to be fair \emph{and} resilient against grieving attacks.
We will introduce a formal definition of \emph{cost fairness} to address the following research questions:

\begin{enumerate}
	\renewcommand{\labelenumi}{RQ \arabic{enumi}}
	\renewcommand{\theenumi}{\arabic{enumi}}
	\renewcommand{\labelenumii}{\arabic{enumi}.\arabic{enumii}}
	\renewcommand{\theenumii}{.\arabic{enumii}}

	\item \label{rq:modeling} How can two-party exchange protocols be modeled so that transaction cost is taken into account?
	\item \label{rq:cost fairness assessment} How can the fairness of two-party exchange protocols be assessed regarding transaction cost?
	\item \label{rq:cost fairness/public blockchains} How to achieve cost fairness for public blockchain-based two party exchange protocols (e.g., FairSwap)?
	\item \label{rq:cost fairness/private blockchains} How to achieve cost fairness for private blockchain-based two party exchange protocols?
\end{enumerate}

In order to introduce the topic and provide the foundations our work bases on, we describe related work in Section \ref{sec:related work}.
Our first contribution, a model for two-party fair exchange protocols, answering RQ \ref{rq:modeling}, is presented in Section \ref{sec:modeling}.
To answer RQ \ref{rq:cost fairness assessment}, as our second contribution, we provide a definition for \emph{partial cost fairness} and \emph{full cost fairness} in Section \ref{sec:cost fairness}.
Our third contribution consists of two theorems, presented in Section \ref{sec:cost fairness achievability}, addressing the achievability of partial cost fairness and full cost fairness, especially in the context of blockchains.
We discuss our contributions and use these theorems to answer RQ \ref{rq:cost fairness/public blockchains} and RQ \ref{rq:cost fairness/private blockchains} in Section \ref{sec:discussion}.
We summarize our work and conclude in Section \ref{sec:conclusion}.

This paper is the extended version of the short paper published by Lohr et al.\ \cite{lohr2022costfairness}.

\section{Related Work}
\label{sec:related work}


Cost fairness has been informally defined by Lohr et al.\ \cite{lohr2020cost}.
Our work provides a formal underpinning for cost fairness that allows for modeling exchange protocols and for assessing them regarding cost fairness.
To this end, we use game theory as a formal framework and apply our model to blockchain-based exchange protocols.

\subsection{Fair Exchange}
\label{sec:related work/fair exchange}

The term \emph{fair exchange} describes the challenge of two or more parties that want to exchange their own goods with the guarantee that, despite absence of mutual trust, no party can gain advantage over the other parties \cite{DBLP:journals/cj/PagniaVG03}.
In this context, several definitions of fairness
have been presented as well as different approaches for designing fair exchange protocols, which claim to ensure a fair exchange (fairness as defined by Asokan \cite{DBLP:books/daglib/0094910}) as long as at least one party follows the fair exchange protocol
\cite{DBLP:conf/crypto/Cleve89,DBLP:conf/podc/Tygar96,DBLP:conf/fc/PagniaJ97,DBLP:conf/sp/BaoDM98,DBLP:conf/ccs/FranklinR97,DBLP:conf/ccs/AsokanSW97,DBLP:books/daglib/0094910}.
It has been shown that it is impossible to achieve fair exchange without involving a trusted third party \cite{even1980relations,pagnia1999impossibility}.
None of the approaches referenced above considers possible transaction cost of involving a trusted third party in an exchange.

\subsection{Game Theory}
\label{sec:related work/game theory}

In general, game theory deals with making strategic decisions when two or more parties interact with each other.
In game theory, the parties are referred to as players, which can choose between and follow different strategies to conduct and finish the interaction in the best way for the individual party by maximizing their payoffs \cite{morris2012introduction}.

Game theory already has been applied to the field of fair exchange
\cite{DBLP:conf/welcom/ButtyanH01,buttyan2000toward,DBLP:journals/cj/PagniaVG03}:
A fair exchange can be interpreted as multi-party game, where the fair exchange protocol can be represented by a game tree and the parties (players) involved in the exchange can choose between different strategies (e.g., ``behave faithfully'' or ``cheat'').
Buttyán and Hubaux introduced game theory as an approach for a formal framework, which can be used to assess and compare different types of fairness \cite{buttyan2000toward}.
While their model can be used to assess fairness of exchange protocols, it lacks the ability to assess other aspects of an exchange protocol such as the cost of involving a trusted third party.

For our work, we adopt and modify the general idea of Buttyán and Hubaux of modeling an exchange protocol using game theory to consider the values of the items to be exchanged as well as the transaction cost, which may arise during an exchange, furthermore additional expenses or revenues such as security deposits, paying or receiving a compensation.

\subsection{Blockchain}
\label{sec:related work/blockchain}

A blockchain is an append-only data structure reflecting a state (e.g., bank account balances, variable values), where each state update is collected into a so-called block, which gets appended to the existing data structure.
All modifications to the data can be verified against a set of rules for allowed modifications
and no single entity can prevent or enforce something related to the data without the support of the majority of blockchain participants \cite{nakamoto2008bitcoin}.
Further research and development has extended the concept to support Turing-complete programs for formalizing modification rules,
usually referred to as \emph{smart contracts}, e.g., in context of the Ethereum blockchain \cite{wood2014ethereum}.
Ethereum smart contracts are computer programs, whose source code is added as bytecode to the blockchain data.
This way, everybody who downloads the Ethereum data can execute the program and verify the results submitted to the network by other participants\footnote{Despite theoretically possible, not every node connected to the Ethereum network does this kind of verification. It is up to the node's administrator to decide if he is willing to invest the computational power and therefore has to pay for the required energy to support the blockchain by enabling the verification mechanisms. Alternatively, a node will accept all blocks of the longest chain of blocks.}.

Several approaches implement a trusted third party for fair exchange using Ethereum smart contracts \cite{DBLP:conf/ccs/DziembowskiEF18,hall2019fastswap,DBLP:conf/ccs/EckeyFS20}.
This is usually done by providing a proof of successful transfer or a proof of misbehavior to the smart contract implementing the trusted third party, who will either forward or pay back the payment if the proof can be verified.
Typically, blockchain-based fair exchange protocols are designed to conduct an exchange of data for money, usually in form of a blockchain-specific financial equivalent, which is often referred to as crypto-currency.
Alternatively, also non-fungible tokens could be exchanged, such as digital ownership representations of physical objects (e.g., house, car).

\section{Modeling Exchange Protocols using Game Theory}
\label{sec:modeling}

In this section, we present our model of an exchange protocol using game theory, building upon the work of Buttyán and Hubaux \cite{buttyan2000toward}.

\subsection{Extensive Game}
\label{sec:modeling/common}

We will build on the notion of an \emph{extensive game}, which can be formalized using as follows:

A \emph{game tree} \cite{morris2012introduction,DBLP:books/daglib/0023252} (see Figure \ref{fig:game tree example} for an example) is a tree that depicts all possible ways to play a game.

\begin{definition}[Game Tree \cite{morris2012introduction}]
	\label{def:game tree}
	A \emph{game tree} $T = (V, E, \mathcal{P}, o, \vv{p})$ is a directed tree with a set of vertices $V$ with root $v_0 \in V$, a set of edges $E \subseteq V \times V$ called moves, a set of $n$ players $\mathcal{P}$, a labeling function $o: V \rightarrow \mathcal{P}$, which labels each non-terminal vertex $v \in V$ with a player $P \in \mathcal{P}$ to \emph{own} $v$ and a labeling function $\vv{p}(v) = (p_{P_1}, ..., p_{P_n})$, which labels each terminal vertex $v \in V$ with an n-tuple of numbers called \emph{payoff}, which defines the individual payoff for each player $P_i$.
\end{definition}

Each vertex $v$ represents a possible state of the game to which $T$ belongs.
Being in a state that is represented by $v \in V$, player $P = o(v)$, $P \in \mathcal{P}$ is responsible to choose the next move, represented by $e = (v, v')$, $e \in E, v' \in V$, leading to a new state $v'$.

\begin{figure}
	\centering
	\includegraphics[width=.8\linewidth]{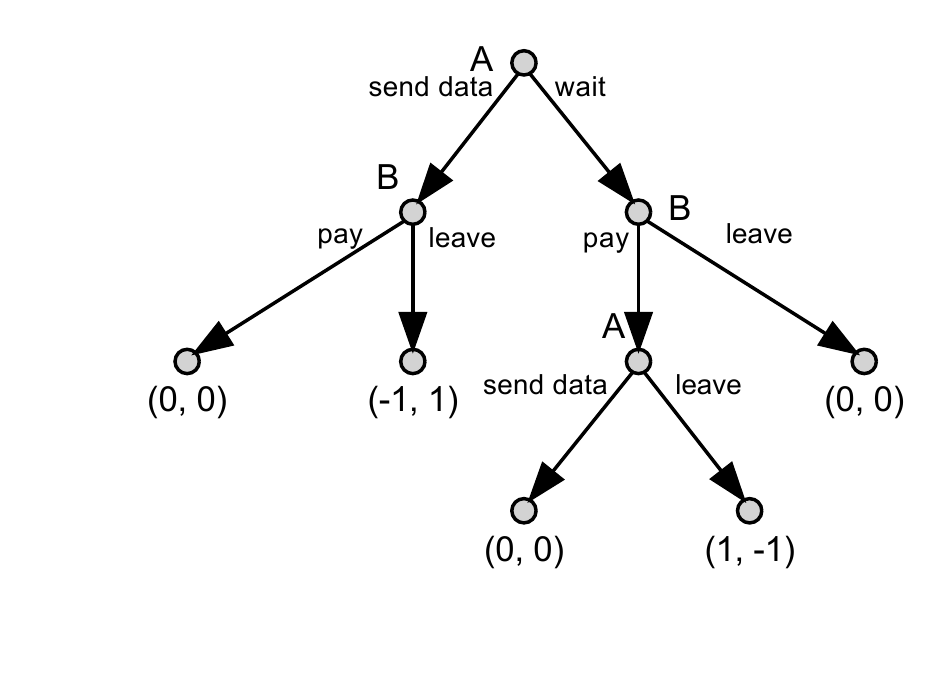}
	\caption{Example of a game tree with players $A$ and $B$ exchanging data for money with different orders of payment and data transfer.}
	\label{fig:game tree example}
\end{figure}

The behavior of players resulting in the selection of the next move in an extensive game is described by a \emph{strategy}.
For simplicity reasons, we only provide a basic definition of a strategy, which covers the aspects required for our work.
For a detailed and more formal definition of strategy we refer to Morris \cite{morris2012introduction}.

\begin{definition}[Strategy]
	A \emph{strategy} $S$ for player $P$ is represented by a partial function called \emph{choice function} $c_P: V \rightarrow V$, which for each $v \in V: o(v) = P$ returns a child $v'$ of $v$ with $(v, v') \in E$ being the next move chosen by $P$ following strategy $S$.
\end{definition}

The set of all available strategies to a player is called \emph{strategy set}:

\begin{definition}[Strategy Set \cite{morris2012introduction}]
	For player $P$ a \emph{strategy set} $\Sigma = \{S_1, ..., S_m\}$ is the set of all possible strategies of $P$.
\end{definition}

Using the previously defined terms, we can now define an extensive game:

\begin{definition}[Extensive Game \cite{morris2012introduction}]
	\label{def:extensive game}
	An \emph{extensive game} is defined as $\Gamma = (T, \mathcal{P}, \{\Sigma_{P_1}, ..., \Sigma_{P_n}\})$ with game tree $T$, set of players $\mathcal{P} = \{P_1, ..., P_n\}$ and their strategy sets $\Sigma_{P_1}, ..., \Sigma_{P_n}$.
\end{definition}

\subsection{Moves of an Extensive Game}
\label{sec:modeling/moves}

Using the terms defined in Section \ref{sec:modeling/common}, we introduce our model of an exchange protocol based on game theory.
For simplicity reasons, we only consider two-party exchange protocols and postpone the expansion to $n$-party exchange protocols to future work.
Similar to Buttyán and Hubaux \cite{buttyan2000toward}, we do not consider the trusted third party to be in the set of players, since we assume that it always behaves deterministically according to the protocol and will never act on its own, only at the instigation of a player.

We assume a two-party exchange with parties $\mathcal{P} = \{A, B\}$ who are interested to exchange their items $\iota_A$ and $\iota_B$.
We assume that $A$ and $B$ agreed on using the exchange protocol $\mathcal{X}$ (we will provide the formal definition of an exchange protocol in Definition \ref{def:exchange protocol}), but neither $A$ nor $B$ can technically be coerced to follow $\mathcal{X}$ during the exchange.
In order to conduct the exchange, $A$ and $B$ can choose their strategies $S_A$ and $S_B$ from their strategy sets $\Sigma_A$ and $\Sigma_B$.
We denote the set of conducted moves of $A$ with $E_A$ and the set of conducted moves of $B$ with $E_B$.

Each move can impact the state of the exchange, e.g., a payment can be conducted or the item (or parts of it, if the item is divisible, e.g., in context of gradual release \cite{DBLP:conf/crypto/Cleve89}) can be handed over between the parties.
We reflect these state changes by a tuple of attributes, which represent the move's effects on the ongoing exchange:

\begin{definition}[Move Attributes]
	\label{def:move attributes}
	Let $e \in E$ be an edge in a game tree $T$ of an extensive game $\Gamma$.
	Let $\mathcal{P} = \{A, B\}$ be the set of players in $\Gamma$.
	Let, w.l.o.g., $A$ be the player conducting $e$.
	We define a tuple $a(e) = (\vv{\rho_e}, \cost_e, \deposit_e, \vv{\compensation_e})$ to be the \emph{move attributes} of $e$,
	where $\vv{\rho_e} = (\rho_e^{A}, \rho_e^{B})$ is a vector of shares of the item transferred to $A$ and $B$ during $e$ with $0 \leq \rho_e^{P} \leq 1$, $P \in \mathcal{P}$,
	$\cost_e \geq 0$ is the transaction cost that has to be paid by $A$ to the trusted third party for conducting $e$,
	$\deposit_e \in \mathbb{R}$ are the funds deposited or retracted by $A$ conducting $e$
	and $\vv{\compensation_e} = (\compensation_e^{A}, \compensation_e^{B})$ with $\compensation_e^{P}$, $P \in \mathcal{P}$ is a vector of the compensations paid out to player $P$ in this move $e$.
\end{definition}

The \emph{item share} $\rho_e^{A}$ refers to the portion of the item $\iota_B$, which is released to $A$ in move $e$.
Indivisible items such as a valuable painting can only be transferred in one piece, in which case $\rho_e^{A} \in \{0, 1\}$.
Divisible items such as money or data can also be transferred in steps, in which case $0 \leq \rho_e^{A} \leq 1$.
Note that $A$ may do a move $e$ that releases an item share $\rho_e^{B}$ to $B$.
The same move $e$ may also trigger that another item share $\rho_e^{A}$ is released to $A$ himself.

The \emph{transaction cost}, denoted with $\cost_e$, describes the fees the party conducting move $e$ has to pay to the trusted third party for conducting move $e$.

In order to enable the trusted third party to punish an unfaithfully behaving party and to compensate a faithfully behaving party, an exchange protocol can require to make a \emph{deposit}, which is managed by the trusted third party.
The total amount of deposit is tracked per party.
A party can change its total deposit in a move $e$ by amount $\deposit_e$ ($\deposit_e > 0$ for depositing, $\deposit_e < 0$ for retracting and $\deposit_e = 0$ for not changing the total amount of the party conducting move $e$).

If $B$ behaves unfaithfully, an exchange protocol can be designed to compensate $A$.
$\compensation_e^A$	denotes the compensation paid to $A$ by the trusted third party in move $e$.

Usually, a trusted third party does not use its own money to pay out compensations.
Instead, the compensation paid out (e.g., to a faithful party) is taken from deposits made before (e.g., from the unfaithful party).
Additionally, for our work we assume the environment, in which the exchange protocol is running, to be a financially closed system.
Therefore, the amount of total compensation paid out can never exceed the total amount of deposits not retracted at the end of the exchange, considering the conducted moves of all players $P_i \in \mathcal{P}$, where $\mathcal{P} = \{A, B\}$:
\begin{equation}
	\label{equation:financially closed environment}
	\sum_{P_i \in \mathcal{P}} \biggl( \sum_{e \in E_{P_i}} \bigl( \deposit_e - \sum_{P_j \in \mathcal{P}} \compensation_e^{P_j} \bigr) \biggr) \geq 0
\end{equation}
Note that a move $e$ conducted by $A$ can cause compensations payouts to $A$ as well as to $B$.

In an exchange of a good for a monetary payment both, the good and the monetary payment, are modeled as items $\iota_\mathit{good}$ and $\iota_\mathit{money}$ of the exchange protocol.
Both goods and money can temporarily be owned by the trusted third party acting as escrow, but only if the good or the money becomes available for the requesting party this is reflected by an item share $\rho > 0$.
E.g., in an exchange using a blockchain-based trusted third party, sending money to the trusted third party does not make it available to one of the parties (therefore $\rho = 0$) while sending unencrypted data to the trusted third party will make it available to everyone (because of the public readability of a blockchain), including the requesting party, therefore $\rho > 0$.

\begin{example}[Move Attributes]
We assume an extensive game with two players $A$ and $B$.
We assume that $A$ is the party conducting the move $e$.
We give three different examples:
\begin{itemize}
	\item $a(e) = ((0, 1), 50, 0, (0, 0))$ -- The move of $A$ makes the item fully available to $B$, charged by the trusted third party with transaction cost $\cost_e = 50$.
	\item $a(e) = ((0.5, 0), 0, 100, (0, 0))$ -- The move of $A$ reveals half of $B$'s item to $A$. $A$ deposits an amount of 100 to the trusted third party that could be used as payment for $B$ in later moves.
	\item $a(e) = ((0, 0), 0, -100, (150, 0))$ -- $A$ withdraws 100 from the funds $A$ deposited with the trusted third party. This is only possible if more than 100 have been deposited by $A$ before and were not used for paying or compensating $B$. Additionally, $A$ retrieves 150 as payment or compensation from the funds deposited by $B$.
\end{itemize}
\end{example}

Even if $A$ and $B$ have agreed on using an exchange protocol $\mathcal{X}$ for their exchange, they usually cannot technically be coerced to conduct a specific move $e \in E$ of $\mathcal{X}$.
Therefore, an exchange protocol $\mathcal{X}$ needs to differentiate between \emph{possible} and \emph{allowed} moves.
In our model, a game tree $T = (V, E, \mathcal{P}, o, \vv{p})$ contains all \emph{possible} moves $e \in E$ for players $P \in \mathcal{P}$.
We label moves \emph{allowed} by an exchange protocol $\mathcal{X}$ to be faithful and all other moves to be unfaithful using the following function:

\begin{definition}[Faithfulness]
	\label{def:faithfulness}
	Let $e = (v, v') \in E$ be an edge in a game tree $T$, $v \in V$ be the parent and $v' \in V$ one of its child nodes.
	We define a total function $\mathit{faithful?}: E \rightarrow \{\mathit{faithful}, \mathit{unfaithful}\}$ that returns for each move $e$ if $e$ is considered to be faithful or unfaithful behavior of player $A = o(v)$.
\end{definition}

Using the definitions presented before, we can now formally define an \emph{exchange protocol} to be a tuple of an extensive game $\Gamma$, a function $a(e)$ that returns move attributes for each move of the game tree of $\Gamma$ and a function $\mathit{faithful?}(e)$ that labels moves to be faithful or unfaithful behavior according to the exchange protocol:

\begin{definition}[Exchange Protocol]
	\label{def:exchange protocol}
	We define an \emph{exchange protocol} $\mathcal{X} = (\Gamma, a, \mathit{faithful?})$ as an extensive game $\Gamma$ together with a function $a(e)$ for retrieving move attributes and a function for determining the faithfulness of a move $\mathit{faithful?}(e)$, $e \in E$ of the game tree of $\Gamma$.
\end{definition}

An exchange protocol $\mathcal{X}$ is called \emph{fair exchange protocol} iff it achieves fairness according to Asokan, who request that in order to achieve fairness, either both parties have to get what they wanted or nobody got anything valuable at the end of the exchange \cite{DBLP:books/daglib/0094910}.

For an exchange protocol $\mathcal{X} = (\Gamma, a, \mathit{faithful?})$, using $\mathit{faithful?}(e)$, $e \in E$ we can classify all available strategies in $\Gamma$ regarding their faithfulness:

\begin{definition}[Faithful and Unfaithful Strategies and Strategy Sets]
	Let $\mathcal{X} = (\Gamma, a, \mathit{faithful?})$ be an exchange protocol.
	We define a strategy $S_A^*$ to be a \emph{faithful strategy} of $A$, if for all possible moves $e = (v, v')$ defined by its choice function $v' = c_A(v)$ it holds that $\mathit{faithful?}(e) = \mathit{faithful}$.
	We define a strategy $S_A^\diamond$ to be an \emph{unfaithful strategy} of $A$, if it is not a faithful strategy of $A$.
	We define the \emph{faithful strategy set} $\Sigma_A^*$ of $A$ as the set of all faithful strategies of $A$.
	We define the \emph{unfaithful strategy set} $\Sigma_A^\diamond$ of $A$ with $\Sigma_A^\diamond = \Sigma_A \setminus \Sigma_A^*$ as the set of all unfaithful strategies of $A$.
\end{definition}

As introduced in Definition \ref{def:game tree}, the quality of a chosen strategy is expressed using its \emph{payoff}.
In an exchange between $A$ and $B$, the payoff for $A$ is everything $A$ received (such as the received shares of $\iota_B$ and received compensations) minus everything $A$ had to give away (such as shares of $\iota_A$, the cost for conducting the exchange, and compensations paid to $B$).
In order to consider the values of the shares of $\iota_A$ and $\iota_B$ for the payoff, we need to introduce a value function that returns the values of the shares of $\iota_A$ and $\iota_B$ in the same unit as the cost or compensation.
However, $A$ and $B$ may have different valuations of the same item $\iota$ and shares of it, therefore $A$ and $B$ each have their own \emph{value function}:

\begin{definition}[Value Function, Valuation]
	\label{def:value function}
	Given a party $A$ and a share $\rho$ of an item $\iota$, the \emph{value function} $v_A(\iota, \rho)$ returns the \emph{valuation} of $A$ regarding the possession of a share of $\rho$ of $\iota$, $0 \leq \rho \leq 1$.
\end{definition}

In a game, the payoff for a player $A$ depends on the strategies chosen by all players of the game:

\begin{definition}[Payoff Function]
	\label{def:payoff function}
	Let $\mathcal{X} = (\Gamma, a, \mathit{faithful?})$ be an exchange protocol with players $A$ and $B$ and let $S_A$ and $S_B$ be their selected strategies.
	Let $c_A(v)$ be the choice function defined by $S_A$ and $c_B(v)$ be the choice function defined by $S_B$.
	Let $E_A$ and $E_B$ be the conducted moves of $A$ and $B$ and $v_t$ be the terminal node after the moves have been conducted.
	Let $a(e) = (\vv{\rho_e}, \cost_e, \deposit_e, \vv{\compensation_e})$ be the move attributes of an edge $e$.
	We define the payoff function $\vv{p}(S_A, S_B)$ such that it labels a terminal vertex $v_t$ in $\mathcal{X}$ with the payoffs $p_A, p_B$ for $A$ and $B$ as follows:
	\begin{align*}
		&(p_A, p_B) = \vv{p}(S_A, S_B) = \vv{p}(v_t) =
		\\ & \biggl( v_A(\iota_B, \sum\limits_{e \in E_A \cup E_B} \rho_e^A) - v_A(\iota_A, \sum\limits_{e \in E_A \cup E_B} \rho_e^B) \\ & + \sum\limits_{e \in E_A} (\compensation_e^A - \deposit_e - \cost_e) + \sum\limits_{e \in E_B} \compensation_e^A~,\\
		& v_B(\iota_A, \sum\limits_{e \in E_A \cup E_B} \rho_e^B) - v_B(\iota_B, \sum\limits_{e \in E_A \cup E_B} \rho_e^A) \\ & + \sum\limits_{e \in E_B} (\compensation_e^B - \deposit_e - \cost_e) + \sum\limits_{e \in E_A} \compensation_e^B \biggl)
	\end{align*}
\end{definition}

Given two strategies $S_A$ and $S_B$, the payoff function $\vv{p}(S_A, S_B) = (p_A, p_B)$ returns the payoff $p_A$ for $A$ for participating in the exchange as well as the payoff $p_B$ for $B$.
The payoff for each player (w.l.o.g. using $A$ as example for now) is calculated by summing up the difference of the value $v_A(\iota_B, \rho_e^B)$ of the item shares received minus the value $v_A(\iota_A, \rho_e^A)$ of the item shares given away (see Definition \ref{def:value function}), plus compensations $\sum_{e \in E_A} \compensation_e^A$ received as a result of moves conducted by $A$, minus deposits $\sum_{e \in E_A} \deposit_e$ made or retracted by $A$ minus the cost $\sum_{e \in E_A} \cost_e$ $A$ has to pay for, plus compensations $\sum_{e \in E_B} \compensation_e^A$ received by $A$ as a result of moves conducted by $B$.

The payoff can be interpreted as financial benefit (or loss) a player experiences participating in an exchange.


If the technical environment cannot force the parties to conduct a next move, a party may leave an exchange at any time.
In this case, it is also not possible to forcefully withdraw money from the leaving party and send it to the faithful party as compensation.
For example, in a blockchain environment, no party can be forced to create new transactions, and withdrawing money from its wallet inevitably requires collaboration.
Since leaving the protocol is not indicated by an explicit action of a party, it has to be assumed by the exchange protocol after a previously defined timeout.
We model the possibility of such an unfaithful leave of an exchange protocol $\mathcal{X}$ with an edge $e_\mathit{leave}$ in its game tree $T$:

\begin{definition}[Unfaithful Leave]
	\label{def:unfaithful leave}
	Let $e_\mathit{leave} \in E$ represent an unfaithful leave, then
	\begin{itemize}
		\item $a(e_\mathit{leave}) = (\vv{0}, 0, 0, \vv{0})$
		\item $faithful?(e_\mathit{leave}) = \mathit{unfaithful}$
	\end{itemize}
\end{definition}

\begin{definition}[Unfaithful Leave At Any Time]
	An exchange protocol $\mathcal{X}$ allows $A$ to \emph{unfaithfully leave at any time}, if for each strategy $S_A \in \Sigma_A$ with $E_A = (e_1, ..., e_n)$
	all strategies $S_A^i$ with $E_A^i = (e_1, ..., e_i, e_\mathit{leave}), 1 \leq i \leq n$ it holds:
	$S_A^i \in \Sigma_A^\diamond$ and also $E_A^0 = (e_\mathit{leave}) \in \Sigma_A^\diamond$.
\end{definition}

\begin{example}[Environment without Unfaithful Leave]
	Assuming a situation in which a shoplifter $B$ can decide to buy or to steal, but if he steals he will definitely be caught by the police.
	When getting caught, he can decide to confess or not to confess, but he cannot leave the police station until he decides either to confess or not to confess.
	This results in a faithful strategy $S_B^f$ with $E_B = (e_\mathit{pay})$ and unfaithful strategies $S_B^\mathit{u1}$ with $E_B = (e_\mathit{steal}, e_\mathit{confess})$ and $S_B^\mathit{u2}$ with $E_B = (e_\mathit{steal}, e_\mathit{not confess})$.
	A strategy $S_B^\mathit{u3}$ with $E_B = (e_\mathit{steal}, e_\mathit{leave})$, in which $B$ leaves the protocol after stealing without the decision of confession is not allowed by the environment and therefore $S_B^\mathit{u3} \notin \Sigma_B$.
\end{example}

Depending on the environment in which an exchange protocol is used, suffering transaction cost might be inevitable.
If transaction cost is non-negligible, we call the exchange protocol to be in an \emph{environment with non-negligible transaction cost}:

\begin{definition}[Environment with non-negligible transaction cost]
	\label{def:environment with non-negligible transaction cost}
	Given an exchange protocol $\mathcal{X}$ represented by game tree $T = (V, E, \mathcal{P}, o, \vv{p})$.
	We define $\mathcal{X}$ to be in an \emph{environment with non-negligible transaction cost} if for all $e \in E \setminus e_\mathit{leave}$ with $a(e) = (\vv{\rho_e}, \cost_e, \deposit_e, \vv{\compensation_e})$: $\cost_e > 0$.
\end{definition}

\section{Cost Fairness}
\label{sec:cost fairness}

Cost fairness has already been informally defined by Lohr et al.\ \cite{lohr2020cost}.
Using the model for exchange protocols described in Section \ref{sec:modeling}, we present a formal definition of two notions of cost fairness.
\emph{Partial cost fairness} provides a guarantee of cost fairness to one of the two parties involved in the exchange while \emph{full cost fairness} provides the guarantee to both parties.

If an exchange protocol $\mathcal{X}$ achieves partial cost fairness in favor of $A$, it will provide the guarantee that regardless whether an actual exchange of items took place the possible benefit (or loss) induced by the exchanged items minus potential cost plus potential compensations received will not lead to a loss for $A$ in total.

\begin{definition}[Partial Cost Fairness]
\label{def:partial cost fairness}
A two-party exchange protocol $\mathcal{X}$ with players $A$ and $B$ achieves \emph{Partial Cost Fairness} in favor of $A$ iff for any strategy $S_B \in \Sigma_B$ for $B$ there exists at least one strategy $S_A \in \Sigma_A^*$ for $A$ such that for $\vv{p}(S_A, S_B) = (p_A, p_B)$ it holds $p_A \geq 0$.
\end{definition}

Applying the definition of partial cost fairness in favor of both parties, $A$ and $B$, an exchange protocol achieves full cost fairness:

\begin{definition}[Full Cost Fairness]
\label{def:full cost fairness}
A two party exchange protocol $\mathcal{X}$ with players $A$ and $B$ achieves \emph{Full Cost Fairness} iff
\begin{itemize}
    \item $\mathcal{X}$ achieves Partial Cost Fairness in favor of $A$ and
    \item $\mathcal{X}$ achieves Partial Cost Fairness in favor of $B$.
\end{itemize}
\end{definition}

Using Definition \ref{def:partial cost fairness} and Definition \ref{def:full cost fairness}, two-party exchange protocols modeled as described in Section \ref{sec:modeling} can be assessed regarding cost fairness as it has been asked for in RQ \ref{rq:cost fairness assessment}.

\section{Achievability of Cost Fairness}
\label{sec:cost fairness achievability}

If w.l.o.g., $B$ cannot leave the exchange protocol without the approval of the trusted third party due to environmental constraints, an exchange protocol could be designed in such a way that $B$ can only leave the exchange protocol if $B$ compensated $A$ for the transaction cost in case that $A$ was behaving faithfully while $B$ was behaving unfaithfully.
This way, an exchange protocol can be designed to always guarantee cost fairness.

\begin{theorem}
	\label{theorem:partial cost fairness impossibility}
	
	Given a two-party exchange protocol $\mathcal{X}$ with parties $A$ and $B$ in an environment with non-negligible transaction cost.
	If $A$ initializes the exchange protocol and $B$ can unfaithfully leave at any time, it is not possible to achieve partial cost fairness in favor or $A$.
\end{theorem}

\begin{proof}[Proof by Contradiction]
	We assume that $A$ initializes $\mathcal{X}$ and $B$ can unfaithfully leave at any time.
	We assume that partial cost fairness in favor of $A$ can be achieved, therefore, according to Definition \ref{def:partial cost fairness}, for any strategy $S_B$ chosen by $B$, there must exist a strategy $S_A$ for $A$ with $\vv{p}(S_A, S_B) = (p_A, p_B)$ where the payoff of $A$ $p_A \geq 0$.
	Since $B$ can leave $\mathcal{X}$ unfaithfully at any time, $B$ can choose a strategy $S_B'$ such that $E_B = (e_\mathit{leave})$.
	According to Definition \ref{def:partial cost fairness} and Definition \ref{def:payoff function}, there has to be a strategy $S_A'$ for $A$ such that
	\begin{align*}
		p_A & = v_A(\iota_B, \sum\limits_{e \in E_A \cup E_B} \rho_e^B) - v_A(\iota_A, \sum\limits_{e \in E_A \cup E_B} \rho_e^A) \\ & + \sum\limits_{e \in E_A} (\compensation_e^A - \deposit_e - \cost_e) + \sum\limits_{e \in E_B} \compensation_e^A \geq 0
	\end{align*}
	Since the only move of $B$ is $e_\mathit{leave}$, $B$ did not share anything to $A$, therefore $v_A(\iota_B, \sum_{e \in E_A \cup E_B} \rho_e^B) = 0$.
	Since $\mathcal{X}$ is a fair exchange protocol, also $A$ did not share anything to $B$, so $v_A(\iota_A, \sum_{e \in E_A \cup E_B} \rho_e^A) = 0$.
	Furthermore, $S_B'$ does not contain any moves causing compensation payouts to $A$, therefore
	$\sum_{e \in E_B} \compensation_e^A = 0$.
	It remains to show that
	\begin{equation}
		\label{equation:partial cost fairness impossibility}
		\sum\limits_{e \in E_A} (\compensation_e^A - \deposit_e - \cost_e) \geq 0   
	\end{equation}
	
	Since $\mathcal{X}$ is assumed to be in an environment with non-negligible transaction cost and $A$ initialized $\mathcal{X}$ with a move $e \neq e_\mathit{leave}$, we know that $\sum_{e \in E_A} \cost_e > 0$.
	Since the only move in $S_B'$ is $e_\mathit{leave}$ with $a(e) = (\vv{0}, 0, 0, \vv{0})$, therefore, in Equation \ref{equation:financially closed environment},
	$\sum_{e \in E_B} \bigl( \deposit_e - \sum_{P \in \mathcal{P}} \compensation_e^{P} \bigr) = 0$.
	Therefore, according to Equation \ref{equation:financially closed environment}, it has to hold that $\sum_{e \in E_A} \bigl( \deposit_e - \sum_{P \in \mathcal{P}} \compensation_e^{P} \bigr) \geq 0$.
	Hence Equation \ref{equation:partial cost fairness impossibility} can never be satisfied. 
	Therefore, for a strategy $S_B'$ with $E_B = (e_\mathit{leave})$ there does not exist such a strategy $S_A'$ such that for $\vv{p}(S_A, S_B) = (p_A, p_B)$ it holds that $p_A \geq 0$, which is a contradiction to the assumption that partial cost fairness can be achieved.
\end{proof}

\begin{theorem}
	\label{theorem:full cost fairness impossibility}
	Given a two-party fair exchange protocol $\mathcal{X}$ with parties $A$ and $B$ using an environment with non-negligible transaction cost.
	If $A$ and $B$ can unfaithfully leave the protocol at any time and moves of $A$ and $B$ are always executed sequentially, it is impossible to achieve full cost fairness.
\end{theorem}

\begin{proof}
	Since moves of $A$ and $B$ are always executed sequentially, either $A$ or $B$ has to initialize $\mathcal{X}$.
	If, w.l.o.g., $A$ initializes the protocol and $B$ can leave unfaithfully, according to Theorem \ref{theorem:partial cost fairness impossibility} partial cost fairness in favor of $A$ cannot be achieved.
	Hence, full cost fairness cannot be achieved.
\end{proof}

\section{Discussion}
\label{sec:discussion}

The main difference of our game-theoretic model of exchange protocols compared with existing models is the consideration of transaction cost and values within the payoff calculation.
We argue why consideration of transaction cost and cost fairness is important for the usability and acceptance of exchange protocols, using blockchain-based exchange protocols as an example.
We also highlight differences regarding transaction cost and cost fairness between public and private blockchains.

\subsection{Game-Theoretic Model of Exchange Protocols}
\label{sec:discussion/model}

In order to answer RQ \ref{rq:modeling}, we developed a model for two-party exchange protocols considering transaction cost.
In contrast to the formal model presented by Buttyán and Hubaux \cite{buttyan2000toward}, in our model of two-party exchange protocols presented in Section \ref{sec:modeling}, we do not consider the actual item but the individual valuations of the items of the parties involved in the exchange for the following reasons:
Game theory generally assumes rational players, which try to maximize their own payoff. If $p_A + p_B < 0$, at least one player will not have any benefit from the exchange, so they rather would not participate in the exchange at all \cite{tao2011formal}.
Looking at individual values $v_A(\iota, \rho^A)$ and $v_B(\iota, \rho^B)$ for an item $\iota$ or a share of it, it is possible that both parties may benefit from an exchange at the same time, if the received item has a higher value for the receiving party than the item that has been passed instead (see \cite{DBLP:books/lib/ShapiroV99}: ``You must price your information goods according to consumer value, not according to your production cost.'').
Furthermore, our model also covers additional financial aspects of an exchange, such as cost (decreasing the benefit) or compensations paid to a party (increasing the benefit).

\subsection{Cost Fairness}
\label{sec:discussion/cost fairness general}

Due to the necessity of the existence of a trusted third party in order to achieve fairness in an exchange \cite{even1980relations,pagnia1999impossibility}, potential transaction cost charged by a trusted third party cannot be avoided when fairness according to Asokan \cite{DBLP:books/daglib/0094910} is required.
For this reason, in Section \ref{sec:cost fairness}, we defined cost fairness, which takes into account transaction cost, but also potential differences in the value of the items to be exchanged and possible compensation payments.
With our definitions of partial cost fairness (Definition \ref{def:partial cost fairness}) and full cost fairness (Definition \ref{def:full cost fairness}) we provide a concept that is applicable for two party-exchange protocols.
With the definitions of cost fairness, we provide a possibility to assess fairness of two-party exchange protocols regarding transaction cost, as asked for in RQ \ref{rq:cost fairness assessment}.

Intentionally, we did not define cost fairness as an extension of fairness, since the concept cannot only be applied for fair exchange protocols but also general exchange protocols (e-commerce as well as in analog world).
As there are protocols, which do not (yet) aim for cost fairness while achieving fairness, it might also be desirable to have a protocol that achieves cost fairness but does not need to achieve fairness.
We suggest that an exchange protocol should try to achieve both, fairness according to \cite{DBLP:conf/ccs/AsokanSW97} and (full) cost fairness.

As long as all parties can be forced to follow the exchange protocol they agreed on and cannot leave it unfaithfully before completing one of the strategies allowed by the protocol, cost fairness can be established by enforcing a compensation payment to the faithful party at the end of the protocol if one party behaves unfaithfully.
If a party can unfaithfully leave the exchange, such a compensation payment directly originating from the unfaithful party cannot be enforced.
To reduce the amount of unilateral cost in such a case, a compensation mechanism can be used, where all parties deposit some money at the beginning of the exchange protocol, which then can be used by the trusted third party to take the amount required to compensate the faithful party from the deposit of the unfaithful party.
However, also the depositing step might raise cost for the faithful party and therefore has to be included in the cost fairness analysis.

\subsection{Application of Cost Fairness to Blockchain-based Fair Exchange Protocols}
\label{sec:discussion/blockchain application}

As the motivation of this work is based on blockchain-based fair exchange protocols, we also want to apply cost fairness to blockchain-based fair exchange protocols.
Since there are fundamental differences between public and private blockchains regarding transaction cost, the assessment of cost fairness has to be done differently for public and private blockchains.

\subsubsection{Cost Fairness in Public Blockchain-based Fair Exchange Protocols}
\label{sec:discussion/blockchain application/public}

In context of public blockchains, transaction cost is accrued in form of fees, which have to be paid per blockchain transaction to incentivize so-called miners in operating and supporting the blockchain infrastructure \cite{wood2014ethereum}.
Therefore, for blockchain-based exchange protocols executed on a public blockchain, having transaction cost is inevitable.
Furthermore, due to the pseudo-anonymity \cite{DBLP:conf/eurosp/BiryukovT19} and the distributed nature of a blockchain, parties involved in the exchange can leave the exchange protocol at any time (by stopping to interact, usually assumed after a timeout defined before the protocol starts).
Therefore, since both parties of a two-party exchange can leave unfaithfully at any time, according to Theorem \ref{theorem:partial cost fairness impossibility} it is not possible to achieve partial cost fairness in favor of the party that has to initialize the exchange protocol.
At least it is possible to achieve partial cost fairness in favor of the second party, if the initializing party is requested to deposit funds during the initialization move.
This compensation can be used by the trusted third party to compensate the other party in case the initializing party behaves unfaithfully.
One example for a blockchain-based two-party fair exchange protocol implementing a compensation mechanism is SmartJudge \cite{wagner2019dispute}.

For an exemplary assessment of a public blockchain-based two-party fair exchange protocol, we take a detailed look at FairSwap, which is designed for the Ethereum blockchain.
The seller initializes the exchange protocol by deploying the smart contract to the blockchain, which is charged with transaction fees of approx.\ 1,050,000 Gas\footnote{approx.\ worth about \FairSwapDeploymentCost~USD as of \FairSwapDeploymentCostDate, see Section \ref{sec:introduction}}.
Since the Ethereum blockchain cannot protect against unfaithful leave, the buyer can leave the protocol right after the seller deployed the contract.
Therefore, partial cost fairness in favor of the seller is not achieved, since the payoff is $(p_\mathit{seller}, p_\mathit{buyer}) = (-1050000 \mathit{Gas}, 0)$.

As shown in Theorem \ref{theorem:partial cost fairness impossibility}, it is not possible to fix FairSwap to achieve full cost fairness.
However, it is possible to reduce cost of the initialization step by creating a \emph{container protocol}, which contains requests a deposit in its initialization move and monitors the behavior of the parties of the contained exchange protocol (e.g., FairSwap) and pays out compensations to the honest party if one party starts to cheat.
Alternatively, state channels \cite{DBLP:conf/ccs/DziembowskiFH18} can be used to execute the protocol off-chain and therefore reduces the amount of blockchain transactions and therefore transaction cost to be paid.

In order to allow for public blockchain-based exchange protocols to achieve cost fairness, a change of the blockchain environment is required in which initializing deposits, such as for initializing a container protocol or opening a state channel, is not charged with transaction cost.

Concluded, answering RQ \ref{rq:cost fairness/public blockchains}, it is not possible to achieve full cost fairness on public blockchains with transaction cost, since (at the current state of art) cost is inevitable and one party has to initialize the protocol, and therefore, according to Theorem \ref{theorem:partial cost fairness impossibility}, partial cost fairness can never be achieved simultaneously in favor of $A$ and $B$.

\subsubsection{Cost Fairness in Private Blockchain-based Fair Exchange Protocols}
\label{sec:discussion/blockchain application/private}

In contrast to a public blockchain, a private blockchain only allows access for well-identified participants.
Therefore, the risk of, e.g., a grieving attack is considerably lower since a party behaving unfaithfully can be punished by getting ignored on future attempts or the access to the private blockchain can be revoked.
Furthermore, a private blockchain does not necessarily come with any means of transaction cost (e.g., Hyperledger Fabric \cite{DBLP:conf/eurosys/AndroulakiBBCCC18}), therefore concepts of cost or money are not an inherent part of a private blockchain.
In this case, means of financial (or comparable) compensations for provided services or items exchanged is in the responsibility of the implementation of the respective smart contract, implementing the exchange.
If means of transaction cost is introduced by such a smart contract, cost fairness can also be assessed like it is done for public blockchains.

However, since the actual operation of the private blockchain network is not necessarily covered by their smart contract applications, also these cost can be taken into account for cost fairness assessment (e.g., cost for servers, internet connection, etc.).
If however (as it is, e.g., with Quorum Blockchain\footnote{Quorum Blockchain --- \url{https://github.com/ConsenSys/quorum}, accessed 2021-12-17}) the private blockchain comes with a financial concept similar to the one of public blockchains, the blockchain network itself could be extended to provide a compensation service of last resort, which takes care about compensation payouts if neither the actual exchange protocol nor superior container protocols are able to provide cost fairness.

Therefore, we have to incorporate operational cost instead of considering transaction cost for the assessment of cost fairness in order to answer RQ \ref{rq:cost fairness/private blockchains}.
Since the choice of the blockchain concept and its implementation is up to the operator(s) of the private blockchain network, they are also free to implement any kind of compensation mechanism, which could be used as compensation of last resort, if inner protocols do not achieve cost fairness.

\section{Conclusion}
\label{sec:conclusion}

In this work, we have introduced our approach on how to model an exchange protocol using notions from game theory (answering RQ \ref{rq:modeling}).
This model can be used as a base for further works for formal analyses of different aspects of two-party exchange protocols.
We used this model to define partial cost fairness and full cost fairness as a desirable property of exchange protocols (answering RQ \ref{rq:cost fairness assessment}).
As major finding, we have shown that cost fairness cannot be achieved on current state-of-the-art blockchains such as Ethereum (answering RQ \ref{rq:cost fairness/public blockchains}).
In private blockchains, which can be designed by the operators, cost fairness can be enabled by allowing for free depositing transactions or even enforced by adding a compensation mechanism as part of the blockchain network (answering RQ \ref{rq:cost fairness/private blockchains}).

In future work, we want to use our model to compare existing blockchain-based two-party exchange protocols regarding different aspects, such as fairness, cost, cost fairness and game-theoretical strategy equilibria  \cite{morris2012introduction}.
Furthermore, we plan to apply state channels to reduce total transaction cost of blockchain-based fair exchange protocols and to allow for a reliable prediction of maximum cost to be covered for the honest party if full cost fairness cannot be achieved, which can be used as a metric for the risk to be taken when joining an exchange.
Related to this, we want to introduce another definition of cost fairness, which considers if the transaction cost of an exchange protocol are guaranteed to stay within the prediction.
We will name this definition \emph{cost fairness with $\epsilon$}, which states if the maximum cost that have to be covered by the honest party if the other party behaves unfaithfully are smaller than $\epsilon$.
The value of $\epsilon$ can then also be used to compare worst-case transaction cost between two exchange protocols.

One drawback of our model is that it is limited to two-party exchanges.
In order to allow a more general usage, we want to extend our model to allow for $n$-party exchanges.


\bibliographystyle{IEEEtran}
\bibliography{IEEEabrv,cost-fairness-extended}

\end{document}